\documentclass[submission,copyright,creativecommons]{eptcs}
\providecommand{\event}{SCSS 2021} % Name of the event you are submitting to
\usepackage{breakurl}             % Not needed if you use pdflatex only.
\usepackage{underscore}           % Only needed if you use pdflatex.

\usepackage{amsmath,amssymb}
\usepackage{nicefrac}
\def\anylabel{\nicefrac{!}{?}}
\usepackage[inline]{enumitem}

\usepackage{amsthm}
\usepackage[utf8]{inputenc}
\usepackage[T1]{fontenc}
\usepackage{url}
\usepackage{comment}

\usepackage{xifthen}
\theoremstyle{plain}
\newtheorem{theorem}{Theorem}
\newtheorem{proposition}{Proposition}

\theoremstyle{definition}
\newtheorem{example}{Example}
\newtheorem{definition}{Definition}
\newtheorem{decisionproblem}{Decision Problem}

\newcommand{\mathconstante}[1]{\ensuremath{{\mathrm{#1}}}}
\newcommand{\Constants}{\ensuremath{\call{C}}}
\newcommand{\Variables}{\ensuremath{\call{X}}}
\newcommand{\fonction}[3][]{%
  \ifthenelse{\isempty{#1}}{%
    \ifthenelse{\isempty{#3}}{%
      \ensuremath{{\mathconstante{#2}}}%
    }{%
      \ensuremath{{\mathconstante{#2}(#3)}}%
    }%
  }{%
    \ifthenelse{\isempty{#3}}{%
      \ensuremath{{\mathconstante{#2}_{#1}}}%
    }{%
      \ensuremath{{\mathconstante{#2}_{#1}(#3)}}%
    }%
  }%
}
\newcommand{\termset}[2][]{{\fonction[#1]{T}{#2}}} 
\newcommand{\gsig}[1]{\termset{#1}}
\newcommand{\vsig}[1]{\termset{#1,\Variables}}
\newcommand{\call}[1]{\ensuremath{\mathcal{#1}}}
\newcommand{\Const}[2][]{\fonction[#1]{Const}{#2}}

\newcommand{\rinput}[2][]{\fonction[#1]{input}{#2}}
\newcommand{\routput}[2][]{\fonction[#1]{output}{#2}}

\newcommand{\trace}[2][]{\fonction[#1]{tr}{#2}}
\newcommand{\mesg}[2][]{\fonction[#1]{mesg}{#2}}
\newcommand{\snd}[2][]{\fonction[#1]{snd}{#2}}
\newcommand{\rcv}[2][]{\fonction[#1]{rcv}{#2}}
\newcommand{\sess}[2][]{\fonction[#1]{sess}{#2}}
\newcommand{\body}[2][]{\fonction[#1]{body}{#2}}

\newcommand{\set}[1]{\ensuremath{{\left\lbrace #1 \right\rbrace}}}
\newcommand{\Var}[2][]{\fonction[#1]{Var}{#2}}

\newcommand{\Cons}[2][]{\fonction[#1]{Cons}{#2}}
\newcommand{\Supp}[2][]{\fonction[#1]{Supp}{#2}}
\newcommand{\tq}{\ensuremath{\,|\,}}
\newcommand{\condset}[2]{\set{#1\tq{}#2}}

\newcommand{\unif}{\ensuremath{\stackrel{?}{=}}}

\newcommand{\penc}[2]{\ensuremath{ \text{\rm enc}(#1,#2)}}
\newcommand{\pdec}[2]{\ensuremath{\text{\rm dec}( #1,#2)}}
\newcommand{\paire}[2]{\ensuremath{\left\langle #1,#2\right\rangle}}
\newcommand{\piun}[1]{\ensuremath{\pi_1(#1)}}
\newcommand{\pideux}[1]{\ensuremath{\pi_2(#1)}}
\newcommand{\msg}[2]{\ensuremath{\text{\rm msg}( #1,#2)}}
\newcommand{\payload}[1]{\ensuremath{\text{\rm payload}(#1)}}

\newcommand{\notark}[1]{}

\def\affiliation#1{#1}
\def\streetaddress#1{#1}
\def\institution#1{#1}
\def\state#1{#1}
\def\email#1{\url{#1}}
\def\city#1{#1}

\title{Implementing Security Protocol Monitors}
\author{Yannick Chevalier\\
  \affiliation{\institution{Irit, Universit\'e Paul Sabatier}\\ 
    \streetaddress{}\city{Toulouse} \\
    \state{France}}\\
\email{ychevali@irit.fr}
\and
Micha\"el Rusinowitch\\
  \affiliation{
  \institution{  Lorraine University, Cnrs, Inria}\\
  \streetaddress{}
  \city{Nancy} \\
  \state{France} \\
  \email{michael.rusinowitch@loria.fr}\\
}}

\def\titlerunning{Implementing Security Protocol Monitors}
\def\authorrunning{Yannick Chevalier and Michaël Rusinowitch}
\begin{document}
\maketitle

  \begin{abstract}
    Cryptographic protocols are often specified by narrations,
    \textit{i.e.}, finite sequences of message exchanges that show the
    intended execution of the protocol.
    Another use of narrations is to describe attacks. % on cryptographic  protocols.
    We propose in this
    paper to compile, when possible, attack describing narrations into
    a set of tests that honest participants can perform to exclude
    these executions.  These tests can be implemented in monitors to
    protect existing implementations from rogue behaviour.
  \end{abstract}

\section{Introduction}

Cryptographic protocols are designed to prescribe message exchanges
between agents in a hostile environment in order to guarantee some
security properties. In particular security properties such as
confidentiality or authentication are violated when there exists an
execution of the protocol in which they do not hold. However it has
often been found that under certain circumstances, and after its
deployment, a protocol failed to adequatly protect its participants.
These circumstances usually involve one or more sessions, and the
participation of a dishonest agent hereafter called the intruder. When
the attack is on a specific implementation of a protocol, its
mitigation usually amounts to fixing this implementation.

However, some attacks are related to the exchanges of messages
prescribed by the protocol, and not in the actual handling of these
messages by participants. In that case, the only recourse is---when
available---to alter the sequence of acceptable messages. This can be
implemented by changing the format of the messages exchanged or by
stopping an execution once it has been detected that the attack may be
under way. We consider in this paper only the second approach, in
which the participants behaviour is altered in order to reject some
possible executions of the protocol.

Let us consider for example the Needham-Schroder Public Key
(\emph{NSPK}) mutual authentication protocol~\cite{NS78} 
%.
% In order to simplify our presentation, we assume that all the
% information is available in the messages' body. Accordingly, we have
% extended the usual representation with the sender, intended recipient,
% and a session id \(S\)
% chosen by the initiator. These informations are necessary for the
% proper routing of messages, but are usually omitted in cryptographic
% protocol descriptions. The rational for that omission is that,
% assuming that the intruder can change messages on the wire, these
% non-protected pieces of information have no consequences on the
% presence or absence of an attack. With this caveat in mind, the NSPK
% protocol is 
%
described by the following sequence of messages between
roles \(A\) and \(B\):
\[
\begin{array}{r@{~:~}l}
  \multicolumn 1l{\ensuremath{A} \textbf{ knows }}& A,B,K_A,K_B,K_A^{-1}\\
  \multicolumn 1l{\ensuremath{B} \textbf{ knows }}& A,B,K_A,K_B,K_B^{-1}\\
  A \to B & \penc{ A , N_A }{K_B} \\
  B \to A & \penc{ N_A ,N_B }{K_A}\\
  A \to B & \penc{ N_B }{ K_B }\\
\end{array}
\]
The attack on this protocol discovered by Lowe~\cite{lowe-nspk} is
described as follows, where  \(I(A)\) denotes the intruder
impersonating the agent \(A\):
\[
\begin{array}{r@{~:~}l}
  A \to I & \penc{ A , N_A }{K_I}\\
  I(A) \to B & \penc{ A , N_A }{K_B}\\
  B \to I(A) & \penc{ N_A , N_B }{K_A}\\
  I \to A & \penc{ N_A , N_B }{K_A}\\
  A \to I & \penc{ N_B }{ K_I }\\
  I(A) \to B & \penc{ N_B }{ K_B} \\
\end{array}
\]
This execution is an attack because \(B\) believes he has participated
in a session with \(A\) whereas \(A\) never exchanged a message with
\(B\).  As is the case in this attack narration, we assume from now on
and without loss of generality~\cite{Lowe99} that in an attack, every
message is sent to or received from the intruder. A fix, proposed
in~\cite{lowe-nspk}, consisting in altering the second message to
include the name of the sender.
% \[
% \begin{array}{r@{~:~}l}
%   \multicolumn 1l{\ensuremath{A} \textbf{ knows }}& A,B,K_A,K_B,K_A^{-1}\\
%   \multicolumn 1l{\ensuremath{B} \textbf{ knows }}& A,B,K_A,K_B,K_B^{-1}\\
%   A \to B & \mesg{A,B,S,\penc{ A , N_A }{K_B}} \\
%   B \to A & \mesg{B,A,S,\penc{ B , N_A , N_B }{K_A}} \\
%   A \to B & \mesg{A,B,S,\penc{ N_B }{ K_B } } \\
% \end{array}
% \]
The only drawback to such fixes is that implementations of the amended
protocol are not interoperable with implementations of the original
protocol. For widely deployed real-life protocols, interoperability
must be maintained and thus the amended version coexists with the
original one for years, leaving open an attack vector for attackers.

Our proposal aims at keeping the original version, but extended with
additional tests. This extension involves the creation of a monitor for the
actions of  honest participants that furthermore may have access to
some secret pieces of information held by these participants. In the
case of Lowe's attack, this access is unnecessary, as it suffices for
\(B\) to check that the message he receives at the third step is equal
to the message sent by \(A\). 

We present in this paper an algorithm to implement a security protocol monitor. 
Given the input messages that participants are willing to share with the
monitor, it basically amounts to computing  the conditions
to be checked in order  to exclude a given narration from
the possible executions of the 
protocol. % We also discuss informally how these tests can be
% implemented in a distributed monitor based on the analysis presented
% in~\cite{CremersRyan}. In particular, we present examples in which the
% attack mitigation does not need multiple participants collaboration,
% and thus are easier to implement.

\paragraph{Related works}
This article is based on the refinement relation between traces
introduced in~\cite{IPL10}. An extension to the case where an attack
can be excluded based on the information in only one session of a
participant has been proposed in~\cite{ASIACCS}. 
% There is also a large number of works that propose transformation schemes 
% for protocol in order to  avoid attacks. These schemes rely in general 
% on tagging  messages with session information or nonces.
% For instance \cite{CarboneG13} the authors introduce synctactic 
% conditions that ensure a protocol behaviour respects sessions 
% in an environment containing active adversaries. They also 
% exhibit a transformation to make a protocol {\em session-respecting}.

By contrast our approach stems from the line of work % we
% pursue here has been 
initiated in~\cite{BellaBM03,ArsacBCC09} where
the authors advocate for the prevention of attacks through detection
and eventually retaliation against the attacker. Also,
\cite{FiazzaPV15} presents in more details an architecture in which the
analysis we present in this paper can be conducted with a better
control on the messages, and also introduce the idea of applicative
firewalls for security protocols. % However, no automated results
% guiding towards the implementation of the actual firewalls are given.

\paragraph{Outline}  We recall in Sec.~\ref{sec:role} how to represent 
protocols and roles and how to implement them as active frames. 
In Sec.~\ref{sec:monitor} we  formally introduce protocol monitors   to 
control messages and manage  knowledge shared  by collaborating agents, 
in order to detect and block attacks. In Sec.~\ref{sec:generating}
we show how to synthesize monitors from tests that can be derived
automatically. We  conclude in Sec.~\ref{sec:conclusion}.

% that are willing to share some information to prevent it. Then we
% present in Section~\ref{sec:attacks} our definition of attack.

\section{Role-based Protocol Specifications}
\label{sec:role}
\subsection{Messages and basic operations}

%\textsl{Terms.} %
We consider an infinite set of free constants \Constants{} and an
infinite set of variables \Variables. For each signature \call{F}
(\textit{i.e.} a set of function symbols with arities), we denote by
\gsig{\call{F}} (resp.  \vsig{\call{F}}) the set of terms over
\(\call{F}\cup{}\Constants{}\) (resp.
\(\call{F}\cup{}\Constants{}\cup\Variables\)).  The former is called
the set of ground terms over \call{F}, while the latter is simply
called the set of terms over \call{F}.  Variables are denoted by \(x\)
and decorations thereof, but for a distinguished subset \((v_i)_{i\in
  \mathbb{N}}\) employed to denote positions in a sequence. Terms are
denoted by \(s\), \(t\), and finite sets of terms are written
\(E,F,\ldots\), and decorations thereof, respectively.
%We abbreviate \(E\cup F\) by \(E,F\), the union \(E\cup\set{t}\) by \(E,t\)
%and \(E\setminus \set{t}\) by \(E\setminus t\).
In a signature \call{F} a \emph{constant} is either a \emph{free
  constant} or a function symbol of arity \(0\) in \call{F}. Given a
term \(t\) we denote by \Var{t} the set of variables occurring in \(t\)
and \Cons{t} the set of free constants occurring in \(t\).  A (ground)
substitution \(\sigma\) is an idempotent mapping from \Variables{} to
\vsig{\call{F}} (\gsig{\call{F}}) and its support \(\Supp{\sigma}=
\condset{x}{\sigma(x)\not=x}\) is a finite set.  The application of a
substitution \(\sigma\) on a term \(t\) (resp.  a set of terms \(E\))
is denoted \(t\sigma\) (resp.  \(E\sigma\)) and is equal to the term
\(t\) (resp. the set of terms \(E\)) where all variables
\(x\in\Supp{\sigma}\) have been replaced by the term \(x\sigma\).

%\textsl{Operations.} %
Terms are manipulated by applying \emph{operations} on them. These
operations are defined by a subset of the signature \call{F} called
the \emph{set of public constructors}. A context \(C[v_1,\ldots,v_n]\)
is a term in which \(\Var{ C[v_1,\ldots,v_n]} \subseteq \set{v_1,
  \ldots , v_n }\), \(\Const{ C[v_1,\ldots,v_n]} = \emptyset \), and
all non-variable symbols are public constructors, including possibly
non-free constant.  We will specify the effects of operations on the
messages and the properties of messages by equations. When the index
\(n\) is clear, we omit the possible variables list and denote
contexts \(C\).
%\textsl{Equational theories.} %
An \emph{equational presentation} \(\call{E}=(\call{F},E)\) is defined
by a set \(E\) of equations \(u=v\) with
\(u,v\in\vsig{\call{F}}\). The \emph{equational theory} generated by
\((\call{F},E)\) on \(\vsig{\call{F}}\) is the smallest congruence
containing all instances of axioms in \(E\) (free constants can also
be used for building
instances)~\cite{DBLP:books/el/RV01/DershowitzP01}.  We write \(s =_{\call{E}}t\) as the congruence relation
between two terms \(s\) and \(t\).  By abuse of terminology we also
call \(\call{E}\) the equational theory generated by the presentation
\call{E} when there is no ambiguity.
%This equational
%theory is introduced in order to specify the effects of operations on the
%messages and the properties of messages.

%\textsl{Deduction systems.}  %
A \emph{deduction system} is defined by a
triple \(( \call{E},\call{F},\call{F}_p)\) where \call{E} is an
equational presentation on a signature \call{F} and \(\call{F}_p\) a
subset of \emph{public} constructors in \call{F}.

\begin{example}{\textrm{\bf Public key cryptography.}}\label{ex:pubkey}
  For instance the following deduction system models public key
  cryptography:
  \[
  \begin{array}{l}
    (  \set{\pdec{\penc{x}{y}}{y^{-1}} = x },\\
    \set{ \pdec{\_}{\_}, \penc{\_}{\_}, {\_^{-1}} }, \\
  \set{\pdec{\_}{\_},\penc{\_}{\_} } ) 
  \end{array}
  \]
  The equational theory is reduced here to a single equation that
  expresses that one can decrypt a ciphertext when the inverse key is
  available. The inverse function \(\_^{-1}\) is not public, as it
  cannot be computed in reasonable time by participants.
\end{example}
\begin{comment}
\begin{example}{\textrm{\bf Messages with header.}}
  Since we have introduced a notation to include extra information to
  the messages actually exchanged, we also give the \texttt{mesg}
  symbol a precise meaning in terms of possible constructions:
  \[
  \begin{array}{l}
    (  \lbrace \snd{\mesg{x,y,z,t}} = y , \rcv{\mesg{x,y,z,t}} = y ,\\
    \sess{\mesg{x,y,z,t}} = z , \body{\mesg{x,y,z,t}} = t \rbrace , \\
    \set{ \mesg{\_,\_,\_,\_}, \snd{\_}, \rcv{\_}, \sess{\_} , \body{\_} }, \\
  \set{\mesg{\_,\_,\_,\_}, \snd{\_}, \rcv{\_}, \sess{\_} , \body{\_} }) 
  \end{array}
  \]
  In other words, given a message \(\mesg{x,y,z,t}\) it is possible to
  retrieve its endpoints, its body, and a session id as defined by an
  underlying transport protocol.
\end{example}

\begin{example}{\textrm{\bf Algebraic properties.}}\label{ex:xor}  Our work is not
    limited to the Dolev-Yao~\cite{dolev83ieee} perfect cryptography
    model, and can accomodate a variety of deduction systems that are used for modelling 
e.g. properties of cryptoprimitives. In the
    latter examples we will in particular employ the \emph{bitwise
      exclusive or}, denoted \(\oplus\), with messages whose bit are
    only \(0\) denoted \(0\), irrespective of their lengths. We associate to 
    this operation the deduction system:
    \[
      \begin{array}{l}
        (\set{ x \oplus x = 0 , x \oplus ( y \oplus z ) = ( x \oplus y ) \oplus z , x\oplus y = y \oplus x , 0\oplus x = x} ,\\
        \set{0,\oplus}, \set{0,\oplus} )\\
      \end{array}
    \]
\end{example}
\end{comment}
\begin{example}{\textrm{\bf Nonce generation.}}
  \label{ex:nonces}
  Nonces are random values that are critical to the analysis of
  cryptography and cryptographic protocols. To give an agent the
  capacity to generate new values, we assume the existence of an
  infinite set of constants \(\Constants_{\mathcal{N}}\) away from
  \(\Constants\) such that each value in this set can be generated:
  \[
    \mathcal{N}=( \emptyset , \Constants_{\mathcal{N}}, \Constants_{\mathcal{N}} )
  \]
  Note this model makes sense only in the case where the attacker is
  only one agent, or a set of information sharing
  agents~\cite{machiavelli}, as an agent cannot otherwise construct
  nonces generated by another, independent, agent.
\end{example}

% \begin{example}{\textrm{\bf Peano integers.}}
%   \textsc{Devenu inutile} Finally, we provide a last example of a
%   deduction system that will be technically useful as it permits to
%   construct an unbounded number of pseudo-constants. We use to this
%   end Peano's definition of natural numbers:
%   \[
%     \mathcal{D_P}=( \emptyset , \set{0,\text{succ}}, \set{0,\text{succ}} )\\
%   \]
% \end{example}

% \textsl{Sequences of terms.} We denote the empty sequence of
% terms by \([]\), and if \(a\) is a term and \(l\) is a sequence of
% \(n\) terms we denote \(a::l\) the sequence of \(n+1\) terms starting
% with \(a\) and completed by \(l\) as usual. If \(s\) is a sequence of
% \(n\) terms, for \(1\le i\le n\) we let \(s_i\) be the \(i\)th term of
% that sequence.  If \(s\) is a sequence of terms, its \emph{length} is
% denoted \(\vert s\vert\) and is the number of elements in \(s\). If \(C\)
% is a term, \(s\) is a finite sequence of \(n\) terms, and \(\theta:
% \Var{C} \to \set{1,\ldots,n}\) we denote \(C\theta s\) the term \(C\)
% in which every occurrence of a variable \(x\) has been replaced by the
% term \(s_{\theta(x)}\).

\textsl{Test systems.} In order to express  verifications 
performed by an agent on received messages we introduce test systems: 
\begin{comment}
Intuitively an operational model for a role has
to reflect the possible manipulations on messages performed by a
program implementing the role. These operations are specified here by
a deduction system \call{D} \(=(\call{E}, \call{F},\call{S})\) where
the set of public functions \call{S}, a subset of the signature
\call{F}, is defined by equations in \call{E}.  Beside defining
function computations, the equations \call{E} specify some properties.
\end{comment}

\begin{definition}{(Test systems)\label{def:unification}}
  Let \(\call{D}\) be a deduction system with an equational theory
  \(\call{E}\).  A \(\call{D}\)-\emph{test system}
  \(S[v_1,\ldots,v_n]\) is a finite set of equations denoted by \((C_i
  \unif{} C'_i)_{i\in\set{1,\ldots,n}}\) with \(\mathcal{D}\)-contexts
  \(C_i[v_1,\ldots,v_n],C'_i[v_1,\ldots,v_n]\).  It is satisfied by a
  substitution \(\sigma\), and we denote by \(\sigma\models{} S[v_1,\ldots,v_n]\), if
  for all \(i\in\set{1,\ldots,n}\) the equality \(C[v_1,\ldots,v_n]_i\sigma
  =_{\call{E}} C'_i[v_1,\ldots,v_n]\sigma\) holds.
\end{definition}
As usual we simply denote a test system \(S\) if the maximal indice
\(n\) is clear from the context.

\subsection{Traces and active frames}
\label{subsec:protocol:description}

%%%%%%%%%%% strand
We model messages with terms. The sequence of messages received and
sent by a principal is a \emph{trace}, and is thus a finite sequence
of labeled messages:
\begin{definition}{(Trace)\label{def:trace}}
  A \emph{trace} is a finite sequence of messages each with label (or
  polarity) \(!\) or \(?\).  
\end{definition}
Messages with label \(!\) (resp. \(?\)) are said to be ``sent''
(resp.``received''). Given a trace \(\Lambda= \anylabel t_1, \ldots ,
\anylabel t_n\) we write \(?\Lambda\) (resp.  \(!\Lambda\)) as a
short-hand for \(?m_1, \ldots , ?m_n\), (resp. \(!m_1, \ldots ,
!m_n\)).  Given a trace \(\Lambda=\anylabel m_1, \ldots , \anylabel
m_n\) we denote by \(\sigma_\Lambda=\set{v_1\mapsto m_1, v_n\mapsto
  m_n}\) the substitution mapping each variable \(v_i\) to the \(i\)th
message occurring in \(\Lambda\). To simplify notation we also denote
by \(C[v_1,\ldots,v_n] \cdot \Lambda\), or more simply \(C\cdot
\Lambda\), the application of the substitution \(\sigma_\Lambda\) on
the context \(C[v_1,\ldots,v_n]\).  Accordingly, we say that a trace
\(\Lambda\) satisfies an equality \(C_1 = C_2\), and denote it by
\(\Lambda\models C_1 \unif C_2\), whenever \(C_1 \cdot \Lambda
=_{\call{E}} C_2\cdot \Lambda\).% ,

\paragraph{Operations on traces} Let \(\Lambda\) be a trace.  We say
that a \(\Lambda\) is positive (\textit{resp.} negative) if all its
labels are \(!\) (\textit{resp.} \(?\)). We denote $\Lambda^?$
(\textit{resp.}  $\Lambda^!$) the subsequence of \(\Lambda\) of terms
labeled with \(?\) (\textit{resp.}  $!$). We denote \(-\Lambda\) the
trace in which all the labels in \(\Lambda\) are inverted. Finally, we
denote \(\rinput{\Lambda}\) (resp. \(\routput{\Lambda}\)) the trace
$-\Lambda^?$ (resp.  $-\Lambda^!$).

\paragraph{Active frames} An \emph{active frame} represents the
actions of a principal participating in a protocol. It is a sequence
of steps, and at step \(i\) the principal either sends a message,
constructed from the messages received at steps \(j<i\), or receives
and message, and accepts it if it satisfies some tests constructed
from it and messages received at steps \(j<i\). To simplify
exposition, at a step \(i\), we call these messages received at steps
\(j<i\) the messages \emph{already known} at step \(i\), or just
\emph{already known} if the step is clear from context. As in the case
of traces, messages sent are labeled \(!v_i\) (and \(v_i\) is an
output variable) and those received are labeled \(?v_i\) (and \(v_i\)
is an input variable). Since the available contexts depend upon the
deduction system, the notion of active frame is also parameterised by
a deduction system.

\begin{definition}{\label{def:active:frame}}
  Given a deduction system \call{D}, a
  \call{D}-\emph{active frame} is a sequence $(T_i)_{1\le i \le k}$
  where
  $$
  T_i=\left\lbrace
    \begin{array}[c]{l@{\hspace*{-1pt}}r}
      !v_i \text{ with } v_i \unif{} C_i[v_1,\ldots,v_{i-1}] & \text{(send)}\\
      \multicolumn 2c {\text{\bf or}}\\
      ?v_i \text{ with } S_i[v_1,\ldots,v_{i}] & \text{(receive)}
    \end{array}
  \right.
  $$
\end{definition}
Without loss of generality and reusing the above notations, a simple 
recursion shows that we can assume that all variables in
\(C_i[v_1,\ldots,v_{i-1}]\) are labeled with \(?\) at a step \(j<i\),
and that all variables in \(S_i[v_1,\ldots,v_i]\) are labeled with
\(?\) at a step \(j\le i\). Without loss of generality, from now on we
assume that this is the case for all the active frames we consider.

\begin{example}{\label{ex:active:frame}}
  The principal $A$ in the description of the NSPK protocol can be
  modeled by an active frame as follows, with the caveat that we have
  renamed the \(v_i\) variables for more clarity:
$$
\begin{array}{l}
  (? x_{N_A}\text{ with }  \emptyset, ? x_{A}\text{ with }  \emptyset,? x_{B}\text{ with }  \emptyset,? x_{K_A}\text{ with }  \emptyset,? x_{K_B}\text{ with }  \emptyset,? x_{K_A^{-1}}\text{ with }  \emptyset,\\
  ! x_{msg_1} \text{ with } x_{msg_1}\unif \penc{\paire {x_{A}} {x_{N_A}} }{x_{K_B}},\\
  ? x_r \text{ with }  \emptyset \\ 
  % \msg B {\penc{\paire {N_A} {N_b} }{K_A}},
  ! x_{msg_2} \text{ with } x_{msg_2}\unif \penc{
      \pideux{\pdec{x_r}{x_{K_A^{-1}}}} }{ x_{K_B} })
\end{array}
$$
\end{example}

Algebraically, by describing a principal's
actions, active frames are partial operations on the set of traces and
map a sequence of messages sent by someone else and accepted to the
sequence of received and sent messages by a principal.  We formalize
these notions as follows:

\begin{definition}{\label{def:evaluation}}
  Let \call{D} be a deduction system with equational theory \call{E}.
  Let $\varphi=(T_i)_{1\le i \le n}$ be an active frame, where the
  $T_i$'s are as in Definition~\ref{def:active:frame}, and where the
  input variables are $?v_{\alpha_1},\ldots,?v_{\alpha_k}$. Let
  $\Lambda$ be a positive trace of length \(k\), \(\theta\) be the
  renaming of variables \(\set{v_{\alpha_j} \mapsto v_j}_{1\le j\le k}\),
  and $S$ be the union of the test systems in $\varphi$.  The
  \emph{evaluation} of $\varphi$ on $\Lambda$ is denoted $\varphi\cdot
  \Lambda$. It is defined, and we say that $\varphi$ \emph{accepts}
  $s$, if $S\cdot s$ is satisfiable. In that case, it is the trace
  $(m_1,\ldots, m_n)$ where:
  $$
  m_i=\left\lbrace
    \begin{array}[c]{ll} 
      ! C_i\cdot \theta\sigma_\Lambda & \text{If }v_i\text{ has label ! in } T_i\\
      ?v_i\cdot \theta\sigma_\Lambda & \text{If }v_i\text{ has label ? in }  T_i\\
    \end{array}\right.
  $$
\end{definition}

% To simplify notations, the application of a
% \call{D}-context $C[x_1,\ldots,x_n]$ on a positive strand
% $s=(!t_1,\ldots,!t_n)$ of length $n$ is denoted $C\cdot s$ and is the
% term $C[t_1,\ldots,t_n]$. 
\begin{example}{\label{ex:frame:application}}
  Let \(\Lambda_A\) be the trace of the principal \(A\) in the
  the specification of the NSPK protocol in the introduction, \(r=\trace{A}\),
  and $\phi_A$ be the active frame of Ex.~\ref{ex:active:frame}. Let
  $M$ be the message $\msg B {\penc{\paire {N_A} {N_b} }{K_A}}$.  We
  have:
  $$ \begin{array}{l}
       \rinput{\Lambda_A}=(! N_A,  ! A,! B,! K_A,! K_B,! K_A^{-1}, 
       !M )
     \end{array} $$
     and $\phi_A\cdot\rinput{\Lambda_A}$ is the trace:
     $$
     \begin{array}{l}
       (?N_A, ? A,?B,?K_A,?K_B,? K_A^{-1}, \\
       !\msg{B}{\penc{\paire {A} {N_A} }{K_B}},\\
       ?M,  !\msg {B} {\penc{ \pideux{\pdec{ \payload{M } }{K_A^{-1}}} }{ K_B }}
  \end{array} 
  $$
  Modulo the equational theory, this trace is equal to:
  $$
  \begin{array}{l}
    (?N_A,  ?A,?B,?K_A,?K_B, ?K_A^{-1},\\
    !\msg{B}{\penc{\paire {A} {N_A} }{K_B}},
    ?M, !\msg {B} {\penc{ N_b }{ K_B }}
  \end{array}
  $$
\end{example}
It is not coincidental that in Ex.~\ref{ex:frame:application} the
traces $\varphi_A\cdot\rinput{\Lambda_A}$ and $\Lambda_A$ are equal as
it means that within the active frame, the sent messages are composed
from received ones in such a way that when someone sends the messages
expected in \(\Lambda_A\), the execution of \(A\) is described by
\(\Lambda_A\). This relation gives us a criterion to define what an
implementation of a trace is.

\begin{definition}{\label{def:implementation}}
  An active frame $\varphi$ is an \emph{ implementation} of a trace
  $\Lambda$ if $\varphi$ accepts $\rinput{\Lambda}$ and
  $\varphi\cdot\rinput{\Lambda}=_{\call{E}}\Lambda$. 
\end{definition} 
If a trace admits an implementation we say this trace is
\emph{executable}.  Conversely we say that a trace \(t\) is an
\emph{execution} of an active frame \(\varphi\) whenever \(\varphi\)
is an implementation of \(t\).

\subsection{Computation of an implementation}
\label{sec:compil:conform}

We present in this section a method, parameterised by the deduction
system \(\mathcal{D}\), to compute an active frame implementing an
executable trace. To build such an implementation we need to compute,
given a message $t$ sent at step \(i\), a \(\call{D}\)-context $C_i$
that evaluates to $t$ when applied to the previously received
messages. This reachability problem is unsolvable in general. Hence we
have to consider systems that admit a reachability algorithm, formally
defined below:

\begin{definition}{}
  Given a deduction system \call{D} with equational theory \call{E}, a
  \emph{\call{D}-reachability algorithm} $\call{A}_{\call{D}}$
  computes, given a positive trace $\Lambda$ of length $n$ and a term
  $t$, a \call{D}-context $\call{A}_\call{D}(s,t) = C[v_1,\ldots,v_n]$
  such that $C\cdot \Lambda =_{\mathcal{E}} t$ iff there exists such a
  context and $\bot$ otherwise.
\end{definition}

For the many theories that admit a reachability algorithm, it can be
employed as an oracle to compute the contexts in sent messages and
therefore to derive an implementation of a trace $s$. We
thus have the following theorem (see a proof in \cite{IPL10}).

\begin{theorem}{\label{theo:executable}}
  If a \call{D}-reachability algorithm exists then it can be
  decided whether a trace $s$ is executable and if so one can compute
  an implementation of $s$.
\end{theorem}

\subsection{Computation of a prudent implementation}
\label{sec:compil:prudent}

An implementation does not necessarily checks the conformity of the
messages with the intended patterns, \textit{e.g.}, the active frame in
Ex.~\ref{ex:frame:application} neither checks that $x_r$ is really an
encryption with the public key $x_{K_A}$ of a pair, nor that the first
argument of the encrypted pair has the same value as the nonce
$x_{N_A}$. % In Section~\ref{sec:compil:conform} we show how to
% compute an active frame when the role specification is executable,
% and ensure in Section~\ref{sec:compil:prudent} that all the possible
% checks are performed.
% We note that having an implementation of a role is of little use
% \emph{w.r.t.} the security analysis of a protocol. For example the
% active frame of Ex.~\ref{ex:active:frame} is an implementation of the
% initiator of the NSPK protocol but it will accept any message from the
% intruder without aborting.

Any of the algorithms proposed so far in the literature for the
compilation of cryptographic protocols would require at least these tests.
% at least require that the role checks that the received message
% contains the nonce sent at the first step.
We now present an algorithm that computes these kinds of
checks for an arbitrary deduction system. It formalizes a check as an
equation between \call{D}-contexts over messages received so far,
including the initial knowledge. For example, and reusing the
notations of Ex.~\ref{ex:active:frame} it computes that upon reception
of the message the initiator must, among other tests, check the
validity of the equation:
\[ 
  \piun{\pdec{x_r}{x_{K_A^{-1}}}} \unif x_{N_A}
\] 
We formalize in the definition below which traces \(\Lambda'\) are
acceptable by an agent expecting a trace \(\Lambda\). We define the
acceptable traces as the refinements of \(\Lambda\), that is traces
\(\Lambda'\) such that every test system accepting \(\Lambda\) also
accepts \(\Lambda'\).

\begin{definition}{\label{def:refinement}}
  Let \(\Lambda,\Lambda'\)
  be two positive traces of identical length.
  We say that \(\Lambda'\)
  \emph{refines} \(\Lambda\)
  if, for any pair of \(\call{D}\)-contexts
  $(C_1, C_2)$ one has
  $ C_1\cdot \Lambda =C_2\cdot \Lambda$ implies
  $C_1\cdot \Lambda'=C_2\cdot \Lambda'$.
\end{definition}
Consider for example the following traces \(\Lambda\) and
\(\Lambda'\):
\[
  \left\lbrace
    \begin{array}{r@{\ensuremath{=}}l}
      \Lambda' & (!\penc{a}{k},!\penc{a}{k'},!k,!\mathbf{k''},!a)\\
      \Lambda & (!\penc{a}{k},!\penc{a}{k'},!k,!\mathbf{k'},!a)\\
    \end{array}
  \right.
\]
since all equalities that can be checked on \(\sigma\)
can be checked on \(\sigma'\).
Two traces \(s,s'\)
that refine one another are \emph{equivalent}.  This definition is an
adaptation to our setting of the classic notion of static
equivalence~\cite{AbadiC06}.

% It is tempting to think of
% refinement as an instantiation of a given strand.  While this
% intuition is often correct, let us note that if a strand \(s\)
% contains the sequence of one term \(! \penc{Na}{K}\),
% then since the term cannot be decrypted using a public context, for
% every term \(t\),
% the strand \(s_t= ! t\)
% is a refinement of \(s\).
% Of course, this is not the case for the strand
% \(s'= ! Na ;! K^{-1} ; !  \penc{Na}{K}\),
% for which only instances of \(Na\) and \(K\) are accepted.
% An implementation $\phi$ of a trace \(\sigma\)
% is \emph{prudent} if \emph{any} sequence of messages accepted by
% $\phi$ refines \(\rinput{\sigma}\).

When the behaviour of a principal is defined by a trace \(\Lambda\),
we expect that any implementation of that principal accepts the trace
\(\rinput{\Lambda}\).  Thus, and as long as only equality tests are
considered, we expect any implementation of the trace \(\Lambda\) to
also accept any refinement of \(\rinput{\Lambda}\). We define a
\emph{prudent implementation} of \(\Lambda\) as an implementation that
only accepts inputs that refine the inputs in \(\Lambda\).

\begin{definition}{\label{def:eq:prudent}}
  An active frame \(\varphi\) is a \emph{prudent implementation} of a
  trace \(\Lambda\) if \(\varphi\) is an implementation of \(\Lambda\)
  and any trace \(\Lambda'\) accepted by \(\varphi\) is a refinement
  of \(\rinput{\Lambda}\).
\end{definition}
As already noted in~\cite{IPL10}, most deduction systems considered in
the context of cryptographic protocols analysis have the property that
it is possible to compute, given a positive trace, a finite set of
context pairs that summarizes all possible equalities. Given a
positive trace $\Lambda$ we denote $P_\Lambda$ the (infinite) set of
context pairs $(C_1,C_2)$ such that $C_1\cdot s= C_2\cdot s$.

\begin{definition}{\label{def:acceptable}}
  A deduction system \call{D} has the \emph{finite basis property} if
  for each positive trace $\Lambda$ one can compute a finite set
  $P_\Lambda^f$ of pairs of \call{D}-contexts such that, for each
  positive trace $\Lambda'$:
  $$
  P_\Lambda \subseteq P_{\Lambda'} \mbox{~~iff~~} P_\Lambda^f \subseteq
  P_{s'}
  $$
\end{definition}
Let us now assume that a deduction system \call{D} has the finite
basis property. There thus exists an algorithm $\call{A}'_{\call{D}}$
that takes a positive trace $\Lambda$ as input, computes a finite set
$P_\Lambda^f$ of context pairs $(C,C')$ and returns as a result the
test system $S_\Lambda:$ $\set{ C \unif C' ~\vert~(C,C')\in
  P_\Lambda^f }$.  For any positive trace \(\Lambda'\) of length $n$,
by definition of $S_\Lambda$ we have that $S_\Lambda\cdot \Lambda'$ is
satisfiable if and only if $s'$ is a refinement of $s$.  We are now
ready to present our algorithm for the compilation of strands into
active frames.

\paragraph{Algorithm} Given a trace \(\Lambda\) and assuming that the
deduction system \(\mathcal{D}\) has a reachability algorithm and the
finite basis property, and let \(\Lambda\) be a trace of length \(n\),
and let us denote \(\Lambda^i\) for \(1 \le i \le n\) the prefix of
length \(i\) of \(\Lambda\), and \(\Lambda(i)\) the \(i\)th element of
\(\Lambda\). We construct a prudent implementation
\(\varphi_\Lambda=(T_i)_{i=1,\ldots,n}\) of \(\Lambda\) as follows:
$$
T_i=\left\lbrace
  \begin{array}[c]{ll}
    ! v_i \text{ with } v_i\unif\call{A}_\call{D}(\Lambda^{i-1},t_i) & 
                                                                       \text{If }\Lambda(i)=! t_i\\
    ? v_i \text{ with } \call{A}_{\call{D}}'(\Lambda^i)& 
                                                         \text{If }\Lambda(i)=?t_i\\
  \end{array}
\right.
$$
By construction we have the following theorem\cite{IPL10}:
\begin{theorem}{\label{th:main}}
  Let \call{D} be a deduction system that has a \call{D}-ground
  reachability algorithm and has the finite basis property.  Then for
  any executable trace \(\Lambda\) one can compute a prudent
  implementation $\varphi_\Lambda$ of \(\Lambda\).
\end{theorem}

\subsection{Protocol implementation and execution}

It is customary to describe a protocol by giving its intended
execution, either using a message sequence chart or an Alice\&Bob
notation. We note that the same notation is also employed to described
\textit{e.g.} attacks on that protocol. Beyond their syntax, the
characteristic of such description is to associate to a generic
principal (a rôle, in the case of a protocol specification, a
participant in the case of an attack description) a trace describing
its actions, and how these actions interact with the other principal
actions.  This association of a participant with a trace is formalised
by a function mapping \emph{strands}~\cite{strand}, \textit{i.e.}
principals, rôles, etc., to traces. We define a protocol to be just
one such mapping.
%\footnote{YC: relevant ? A simple way to think of it is that a
%  principal is defined by a definite set of pieces of data, such as
 % keys, names, etc., a trace is a program specification, an active
 % frame is a program, and a strand is a process id.}

\begin{definition}{(Protocol)\label{def:protocol}}
  A \emph{protocol} is a couple \(P=(\Xi_P,\trace[P]{})\) where
  \(\Xi_P\) is a finite set of strands and \(\trace[P]{}\) maps
  \(\Xi_P\) to the set of traces. 
\end{definition}
When a protocol is intended to be a protocol specification, we refer
to strands as the rôles of that protocol (\textit{e.g.} the rôles
\(A\) and \(B\) in the NSPK protocol.  A strand \(\xi\) is
\emph{positive} in a protocol \(P\) if \(\trace[P]{\xi}\) is a
positive trace.  

In the preceding definition the function \(\trace[P]{}\) prescribes
for each role \(\xi\in\Xi_P\) the sequence of actions to be performed
by an agent playing this role in any protocol instance.  In the
following, when there is no ambiguity in the considered protocol, we
identify a strand $\xi $ with its trace $\trace{\xi}$.

We have worked so far on the implementation of the trace of a role in
a protocol, but the definitions lift to the level of an implementation
of a protocol as follows.

\begin{definition}{(Protocol implementation)\label{def:protocol:implementation}}
  An \emph{implementation of a protocol \(P=(\Xi _P,\trace[P]{})\)} is
  a couple \((\Xi_P,\Phi_P)\) where \(\Phi\) maps each role
  \(\xi\in\Xi_P\) to an active frame such that \(\Phi(\xi)\) is an
  implementation of \trace[P]{\xi}. It is \emph{prudent} if moreover for
  each \(\xi\in\Xi_P\),
  % the active frame
  \(\Phi(\xi)\) is a prudent implementation of \trace[P]{\xi}.
\end{definition}

From now on we consider only prudent implementations of protocols,
\textit{i.e.} implementations whose execution is %always by definition
a refinement of the protocol specification.

\begin{definition}{(Protocol execution)\label{def:protocol:execution}}
  Let \(P=(\Xi_P,\Phi_P)\) be a protocol implementation. A triple
  \(E=(\Xi_E,\trace[E]{},R_E)\) where:
  \begin{enumerate}
  \item \(\Xi_E\) is a set of strands away from \(\Xi_P\);
  \item \((\Xi_E,\trace[E]{})\) is a protocol;
  \item \(R_E:\Xi_E\to \Xi_P \cup \set{I}\).
  \end{enumerate}
  is a \emph{protocol execution} of \(P\) if, for each
  \(\xi\in\Xi_E\), if \(R_E(\xi)\neq I\) then \trace[E]{\xi} is an
  execution of \(\Phi_P(R_E(\xi))\).
\end{definition}
The strand \(I\) denotes an \emph{Intruder} who does not necessarily
follows the directions prescribed by the protocol.  A protocol
execution is \emph{honest} if \(R_E(\Xi_E)\subseteq \Xi_P\). Strands
in \(\Xi_E\) are called the \emph{participants} of the protocol
execution \(E\). The function \(R_E\) maps each (honest) participant
to its rôle in the protocol.

%%%%%%%%%%%%%%%%%%%%%%%%%%%%%%%%%%%%%%%%%

\section{Protocol monitor}
\label{sec:monitor}

To mitigate an attack on a protocol, a monitor has to coordinate the
participants to detect and stop an instance of a known flaw.  This
coordination is built according to data the participants are willing
to share to prevent the attack.  Our monitor construction relies on
the description of the data the participants are willing to share, a
description of the attack, and a description of the expected behaviour
of the participants, and we compute tests (when possible) to
distinguish an instance of the attack from the normal execution.

In Def.~\ref{def:protocol:monitor}, for each participant \(A\),
\(\trace[M]{A}\) contains the same inputs as \(\trace[P]{A}\), and the
messages sent in \(\trace[M]{A}\), are the pieces of data shared by
\(A\) with the monitor.

\begin{definition}{(Protocol Monitor)\label{def:protocol:monitor}}
  Let \(P=(\Xi_P,\trace[P]{})\)
  and \(M=(\Xi_M,\trace[M]{})\)
  be two protocols. We say that \emph{\(M\) is a monitor for \(P\)} if 
  \begin{enumerate*}
  \item \(\Xi_M= \Xi_P\);
  \item \(M\) is executable; and
  \item For each \(\xi\in\Xi_M\) we have
    \(
      \rinput{\trace[M]{\xi}} = \rinput{\trace[P]{\xi}}
    \).
  \end{enumerate*}
\end{definition}

% \textsc{trace sans output pourne pas monitorer un role}
% Bonjour Yannick
% dans la proposition 1  on n'a pas défini
% conform execution je crois.
% Aussi faut il remplacer  \(\rinput{\trace[E]{e}}\) par \(\rinput{e}\) dans la preuve? 
% Sinon est ce evident que les inputs 
% variables de \((\phi)I_p(r)\)  sont instanciées par des termes clos ? 
% C'est peut etre la definition de conform execution

\begin{proposition}{\label{prop:unique:input}}
  Let \(P=(\Xi_P,\trace[P]{})\) be a protocol,
  \(M=(\Xi_P,\trace[M]{})\) be a monitor of \(P\),
  \(I_X=(\Xi_P,\Phi_X)\) be any implementation of \(X\in\set{P,M}\),
  and \(E=(\Xi_E,\Phi_E,R_E)\) be an honest execution of \(I_P\).

  If \(I_P\) is prudent and \(\Phi_{M}(R_E(\xi))\) accepts
  \(\rinput{\trace[E]{\xi}}\), then \(\Phi_{M}(R_E(\xi))\cdot
  \rinput{\trace[E]{\xi}}\) is a refinement of \trace[M]{R_E(\xi)}.
\end{proposition}

\begin{proof}
  Assume there exists \(\xi_R \in\Xi_M\) and \(\xi_e\in\Xi_E\) with
  \(R_E(\xi_e)=\xi_R\) such that \(\Phi_{M}(\xi_R)\cdot
  \rinput{\trace[E]{\xi_e}}\) is not a refinement of \trace[M]{\xi_R}.
  That is, there exists pairs of contexts \(C_1,C_2\) such that
  \(\trace[M]{\xi_r}\models C_1 = C_2\) but \(\Phi_{M}(\xi_R)\cdot
  \rinput{\trace[E]{\xi_e}}\not\models C_1 = C_2\). Without loss of
  generality we can assume that \(C_1,C_2\) are built upon the input
  variables of \(\Phi_{M}(R_E(\xi))\), that is, with \(\theta:
  \set{v_{i_j}\mapsto v_j}_{1\le j\le k}\), where \(i_j\) is the
  \(j\)th input step of \trace[M]{\xi_R}:
  \[
    \left\lbrace
      \begin{array}{l}
        \rinput{\trace[M]{\xi_R}}\models C_1\theta = C_2\theta\\
        \rinput{\Phi_{M}(\xi_R)\cdot \rinput{\trace[E]{\xi_e}}}\not\models C_1\theta = C_2\theta\\
      \end{array}
    \right.
  \]
  Since \(I_M\) is an implementation of \(M\), by definition the
  second assertion is equal to \(\rinput{\trace[E]{\xi_e}}\not\models
  C_1\theta = C_2\theta\). By definition of a monitor, we have
  \(\rinput{\trace[M]{\xi_R}}=\rinput{\trace[P]{\xi_R}}\). Thus, we
  have:
  \[
    \left\lbrace
      \begin{array}{l}
        \rinput{\trace[P]{\xi_R}}\models C_1\theta = C_2\theta\\
        \rinput{\trace[E]{\xi_R}}\not\models C_1\theta = C_2\theta\\
      \end{array}
    \right.
  \]
  Hence \(\rinput{\trace[E]{\xi_R}}\) is not a refinement of
  \(\rinput{\trace[P]{\xi_R}}\), and thus \(\Phi_{P}(\xi_R)\) cannot
  be a prudent implementation of \trace[P]{\xi_R}.
\end{proof}

\begin{definition}{(Execution Log)\label{def:log:execution}}
  Let \(P=(\Xi_P,\trace[P]{})\) be a protocol, \(I_P=(\Xi_P,\Phi_P)\)
  be an implementation of \(P\), \(E=(\Xi_E,\trace[E]{},R_E)\) be an
  execution of \(I_P\), \(<_E\) be an arbitrary total order on the
  participants, and \(I_M=(\Sigma_P,\phi_{M})\) be an implementation
  of a monitor \(M\) of \(P\).  The \emph{execution log} of \(E\) for
  monitor \(M\) is the concatenation of the traces:
  \[
    \routput{\phi_M(R_E(\xi_e)) \cdot \rinput{\trace[E]{\xi_e}}}
  \]
  for \(\xi_e\in\Xi_E\) such that \(R_E(\xi_e) \neq I\) in the
  increasing order with respect to \(<_E\).
\end{definition}

\begin{proposition}{\label{prop:unique:execution:log}}
  Let \(P=(\Xi_P,\trace[P]{})\) be a protocol, \(I_P=(\Xi_P,\Phi_P)\)
  be an implementation of \(P\), \(E=(\Xi_E,\trace[E]{},R_E)\) be an
  execution of \(I_P\), \(<_E\) be an arbitrary total order on the
  participants, and \(I_M=(\Xi_P,\Phi_{M})\) be an implementation of a
  monitor \(M\) of \(P\).  Then there exists a unique execution log of
  \(E\) for \(M\).
\end{proposition}

\begin{proof}
  For each \(\xi_e\in\Sigma_E\) let \( \varphi_e=\Phi_M(R_E(\xi_e)\)
  be the active frame executed by \(\xi_e\), and let
  \(\text{in}_e=\rinput{\trace[E]{\xi_e}}\) denote the messages
  received by \(\xi_e\).  Since sent messages are built by a context
  over preceding messages an easy recurrence shows that the value of
  each message in \(\varphi_e\cdot \text{in}_e\) is uniquely defined
  by the values in \rinput{\trace[E]{\xi_e}}.  Thus
  \(\routput{\varphi_e\cdot \text{in}_e}\) is uniquely defined for
  each participant \(\xi_e\in\Xi_E\).  Since the order \(<_E\) is
  total the concatenation of these traces is unique.
\end{proof}

Since the ordering \(<_E\) is arbitrary, we usually omit any reference
to it. By Prop.~\ref{prop:unique:execution:log} the execution log
depends only on the monitor, not on its implementation. Accordingly we
denote it \(\log_{I_P,M}(E)\).  Assuming there exists a
\call{D}-reachability algorithm, it is possible to compute an
implementation of \(M\) whenever \(M\) is executable. Thus given a
monitor \(M\) the function \(\log_{I_P,M}(E)\) can be effectively
computed.

\section{Generating an attack-preventing monitor}
\label{sec:generating}
%%%%%%%%%%%%%%%%%%%%%%%%%%%%%%%%%%%%%%%%%
\subsection{Attack presentation}
\label{subsec:attack:presentation}

In our setting attacks are simply specified as protocol executions 
without  reference to any violated security
property.  The flexibility entailed by this choice however implies
that, in order to prevent the given execution, one also has to provide
what should have been the correct execution for the subset of
participants involved in the attack. This setting leads to the
definition of an attack presentation sharing the same set of
participants playing the same roles, but having different traces.

\begin{definition}{(Attack definition)\label{def:attack:definition}}
  Let \(I_P=(\Xi_P,\Phi_P)\) be a protocol implementation. An
  \emph{attack definition} on \(I_P\) is a tuple
  \((\Xi_E,\trace[A]{},\trace[N]{},R_E)\) such that
  \((\Xi_E,\trace[A]{},R_E)\) is an execution of \(I_P\) and
  \((\Xi_E\setminus R_E^{-1}(I),\trace[N]{},R_E)\) is an honest
  execution of \(I_P\).
\end{definition}

Given an attack definition \((\Sigma_E,\trace[A]{},\trace[N]{},R_E)\),
\((\Sigma_E,\trace[A]{},R_E)\) refers to the \emph{attack execution}
while \((\Sigma_E,\trace[N]{},R_E)\) refers to the \emph{normal
  execution} of the protocol expected for the honest participants
involved. Though this is not enforced by the definition and not needed
in the rest of this paper, it is expected that the initial segments of
the traces corresponding to the initial knowledge and the generation
of nonces should be the same for each participant in the two
executions.

\begin{definition}{(Attack presentation)\label{def:attack:presentation}}
  Let \(P=(\Xi_P,\trace[P]{})\) be a protocol, \(I_P=(\Xi_P,\Phi_P)\)
  be an implementation of \(P\), \(M=(\Xi_P,\trace[M]{})\) be a
  monitor of \(P\), and
  \(\mathcal{A}=(\Xi_E,\trace[A]{},\trace[N]{},R_E)\) be an attack
  definition on \(I_P\).
  Then the \emph{presentation of \call{A} to \(M\)}
  is the couple
  \((\log_{I_P,M}((\Xi_E\setminus R_E^{-1}(I),\trace[A]{},R_E)),\log_{I_P,M}((\Xi_E\setminus R_E^{-1}(I),\trace[N]{},R_E)))\)
\end{definition}

\paragraph{Detectable attacks.}
We note that the two traces in an attack presentation may be
equivalent. In this case, no test performed by the monitor could
enable it to distinguish between the normal and the attack execution,
and the latter would not be preventable. We say that an attack
\call{A} is \emph{detectable} by the monitor \(M\) if its presentation
\((\Lambda,\Lambda')\) to \(M\) is such that \(\Lambda\) and
\(\Lambda'\) are not equivalent.

% either \(T_1\)
% does not refine \(T_2\)
% or \(T_2\)
% does not refine \(T_1\).
% Before proceeding to the prevention of an attack,
This definition leads to the problem of deciding whether an attack is
detectable by a monitor.

\begin{decisionproblem}{\ensuremath{\mathop{\it AttackDetectability}_{\mathcal{D}}(s,s')}}
  \begin{description}
  \item[Input:] The presentation \((\Lambda,\Lambda')\) of an attack \call{A} on the protocol
    implementation \(I_P\) with a monitor \(M\);
  \item[Output:] \textsc{yes} if \call{A} is detectable by \(M\)
  \end{description}  
\end{decisionproblem}

This problem is related to the classic static equivalence problem by
the following theorem, proved in the appendix.
\begin{theorem}\label{theo:decidability}
  Let \(\mathcal{D}\)
  be a deduction system, and \(\mathcal{N}\)
  be the deduction system of Ex.~\ref{ex:nonces}.
  Then \ensuremath{\mathop{\it
      AttackDetectability}_{\mathcal{D}\cup\mathcal{N}}} on strands
  that do not contain symbols of \(\Constants_{\mathcal{N}}\) is
  polynomial-time reducible to \ensuremath{\mathop{\it
      StaticEquivalence}_{\mathcal{D}}}.
\end{theorem}
The latter \ensuremath{\mathop{\it StaticEquivalence}_{\mathcal{D}}}
decision problem is well-studied and in most cases of deduction
systems of interest was found to be decidable, which implies that the
\ensuremath{\mathop{\it
    AttackDetectability}_{\mathcal{D}\cup\mathcal{N}}} problem is also
decidable for most deduction systems of interest.

\subsection{Monitor Synthesis}
\label{subsec:monitor:synthesis}

In our setting an attack definition relies on humans to specify also
the intended execution, but this execution is not present when
searching whether a concrete execution is an attack. Thus we need to
synthesize tests that will detect whether an execution is an attack
by relying solely on the contents of the actual execution.  % This
% synthesis can be achieved along the lines followed for the compilation
% of cryptographic protocols when the deduction system has the finite
% basis property (Def.~\ref{def:finite:basis:property}).

Let \((\Lambda,\Lambda')\) be a detectable attack presentation. By
definition there exists at least one equation \(C_1\unif{} C_2\)
either in \(P_\Lambda^f\) or in \(P_{\Lambda'}^f\) that is not
satisfied by the other trace. We add it to the tests of the monitor.
If the equation is in \(P_s^f\) the monitor interrupts the protocol if
it is not satisfied, whereas if it is in \(P_{s'}^f\) the monitor
interrupts the protocol if it is satisfied.

\section{Attack detection in practice}
\label{sec:distributed:implementation}

% In this section we assume that a set of equality and disequality tests
% have been computed (see above) \textit{w.r.t.} to a given set of
% collaborating agents who are willing to share some knowledge,
% \textit{i.e.} values constructed from the messages exchanged. The
% difficulty in practice is to enable the construction of the tests as
% our method is blind to the origin of a given value. For example, in
% the case of the NSPK protocol, it suffices that the agents playing the
% roles of \(A\) and \(B\) share each his or her value for the nonce
% \(Na\) to detect the attack, as these two values will be distinct in
% the case of an attack. While in this case no additional construction
% is needed, we note that in general, even the construction of public
% context may have to be split among different agents. 

% We plan to address this issue in future work by treating separately
% the different cases underlying the principles described
% in~\cite{DBLP:conf/csfw/MilnerCYR17} for sharing information when
% designing a protocol. 

We present in this section a simple example, the ISO/IEC 9797-1
protocol, especially its manual authentication mecanism
7a described in~\cite{these-robin}. The normal run of the protocol is, after
a human user sent \(D\)
and \(R\) to the two devices \(A\) and \(B\):
\[
\begin{array}{r@{~:~}l}
  \multicolumn 1l{\ensuremath{A} \textbf{ knows }}& A,B,D,R\\
  \multicolumn 1l{\ensuremath{B} \textbf{ knows }}& A,B,D,R\\
  A \to B & h(A,D,k_A,R)\\
  B \to A & h(B,D,k_B,R)\\
  A \to B & k_A\\
  B \to A & k_B\\
\end{array}
\]
A dishonest participant \(i\) can sent back the first message directly
to a honest participant \(a\) willing to play the rôle \(A\), and
completely impersonate \(B\) during the session:
\[
\begin{array}{r@{~:~}l}
  a \to I & h(A,D,k_A,R)\\
  I \to a & h(A,D,k_A,R)\\
  a \to I & k_A \\
  I \to a & k_A\\
\end{array}
\]
We let \(P=(\set{A,B},\trace[P]{})\) be the definition of the
protocol, \(E=(\set{a,i},\trace[E]{}, \set{a\mapsto A, i\mapsto I})\)
be the execution of the protocol \(P\) describing the attack, and
\(M=(\set{A,B},\trace[M]{})\) with:
\[
  \left\lbrace
  \begin{array}{ll}
    \trace[M]{A} &= (?A,?B,?D,?R,?k_A,!h(A,D,k_A,R),?h(A,D,k_B,R),!h(A,D,k_B,R),?k_B)\\
    \trace[M]{B} & = \rinput{\trace[P]{B}}\\
  \end{array}
\right.
\]
The implementation of this monitor would be:
\[
  \begin{array}{rl}
  \Phi_M(A) = & (? v_1 ; ? v_2 ; ? v_3 ; ? v_4 ; ! v_6 \text{ with } \set{x_6 \unif{} h(x_1,x_3,x_5,x_4)}; ? v_7  , ! v_8 \text{ with } \set{x_8 \unif{} x_7} ) \\
  \end{array}
  % \set{ x_6 \unif{} \mesg{x_1,x_2,Si,h(x_1,x_3,x_5,x_4)}})
\]
The two logs for the regular execution and the attack are respectively, with this implementation:
\[
  \left\lbrace
    \begin{array}[c]{ll}
      (  ! h(A,D,k_A,R) ; !h(B,D,k_B,R))&\hspace*{1cm}\text{(normal)}\\
      (  ! h(A,D,k_A,R) ; !h(B,D,k_A,R))&\hspace*{1cm}\text{(attack)}\\
    \end{array}
  \right.
\]
and the test \(x_1 \unif{} x_2\) is satisfied by the log of attack
trace but not by the log of the normal execution.  Thus the monitor
can reject the attack from the log when this equality is satisfied. A
more robust monitor would send the last two messages \(k_A\) and
\(k_B\) as in that case we know of no other attack even when keys are
guessable or the hash function is weak~\cite{robin:csf17}.

\section{Conclusion}
\label{sec:conclusion}
In future work we plan to generate monitor implementations from
several roles, and to study test simplification techniques for
efficiency. We also need to extend the monitor construction in order
to protect a protocol from all the refinements of an attack.

% \bibliography{main}

\begin{thebibliography}{10}
\providecommand{\bibitemdeclare}[2]{}
\providecommand{\surnamestart}{}
\providecommand{\surnameend}{}
\providecommand{\urlprefix}{Available at }
\providecommand{\url}[1]{\texttt{#1}}
\providecommand{\href}[2]{\texttt{#2}}
\providecommand{\urlalt}[2]{\href{#1}{#2}}
\providecommand{\doi}[1]{doi:\urlalt{http://dx.doi.org/#1}{#1}}
\providecommand{\bibinfo}[2]{#2}

\bibitemdeclare{article}{AbadiC06}
\bibitem{AbadiC06}
\bibinfo{author}{Mart{\'{\i}}n \surnamestart Abadi\surnameend} \&
  \bibinfo{author}{V{\'{e}}ronique \surnamestart Cortier\surnameend}
  (\bibinfo{year}{2006}): \emph{\bibinfo{title}{Deciding knowledge in security
  protocols under equational theories}}.
\newblock {\sl \bibinfo{journal}{Theor. Comput. Sci.}}
  \bibinfo{volume}{367}(\bibinfo{number}{1-2}), pp. \bibinfo{pages}{2--32},
  \doi{10.1016/j.tcs.2006.08.032}.

\bibitemdeclare{inproceedings}{ArsacBCC09}
\bibitem{ArsacBCC09}
\bibinfo{author}{Wihem \surnamestart Arsac\surnameend},
  \bibinfo{author}{Giampaolo \surnamestart Bella\surnameend},
  \bibinfo{author}{Xavier \surnamestart Chantry\surnameend} \&
  \bibinfo{author}{Luca \surnamestart Compagna\surnameend}
  (\bibinfo{year}{2009}): \emph{\bibinfo{title}{Validating Security Protocols
  under the General Attacker}}.
\newblock In \bibinfo{editor}{Pierpaolo \surnamestart Degano\surnameend} \&
  \bibinfo{editor}{Luca \surnamestart Vigan{\`{o}}\surnameend}, editors: {\sl
  \bibinfo{booktitle}{Foundations and Applications of Security Analysis,
  {ARSPA-WITS} 2009, York, UK, March 28-29, 2009, Revised Selected Papers}},
  {\sl \bibinfo{series}{Lecture Notes in Computer Science}}
  \bibinfo{volume}{5511}, \bibinfo{publisher}{Springer}, pp.
  \bibinfo{pages}{34--51}, \doi{10.1007/978-3-642-03459-6\_3}.

\bibitemdeclare{inproceedings}{BellaBM03}
\bibitem{BellaBM03}
\bibinfo{author}{Giampaolo \surnamestart Bella\surnameend},
  \bibinfo{author}{Stefano \surnamestart Bistarelli\surnameend} \&
  \bibinfo{author}{Fabio \surnamestart Massacci\surnameend}
  (\bibinfo{year}{2003}): \emph{\bibinfo{title}{A Protocol's Life After
  Attacks...}}
\newblock In \bibinfo{editor}{Bruce \surnamestart Christianson\surnameend},
  \bibinfo{editor}{Bruno \surnamestart Crispo\surnameend},
  \bibinfo{editor}{James~A. \surnamestart Malcolm\surnameend} \&
  \bibinfo{editor}{Michael \surnamestart Roe\surnameend}, editors: {\sl
  \bibinfo{booktitle}{Security Protocols, 11th International Workshop,
  Cambridge, UK, April 2-4, 2003, Revised Selected Papers}}, {\sl
  \bibinfo{series}{Lecture Notes in Computer Science}} \bibinfo{volume}{3364},
  \bibinfo{publisher}{Springer}, pp. \bibinfo{pages}{3--10},
  \doi{10.1007/11542322\_2}.

\bibitemdeclare{article}{IPL10}
\bibitem{IPL10}
\bibinfo{author}{Yannick \surnamestart Chevalier\surnameend} \&
  \bibinfo{author}{Micha{\"{e}}l \surnamestart Rusinowitch\surnameend}
  (\bibinfo{year}{2010}): \emph{\bibinfo{title}{Compiling and securing
  cryptographic protocols}}.
\newblock {\sl \bibinfo{journal}{Inf. Proc. Lett.}}
  \bibinfo{volume}{110}(\bibinfo{number}{3}), pp. \bibinfo{pages}{116--122},
  \doi{10.1016/j.ipl.2009.11.004}.

\bibitemdeclare{inproceedings}{robin:csf17}
\bibitem{robin:csf17}
\bibinfo{author}{St{\'{e}}phanie \surnamestart Delaune\surnameend},
  \bibinfo{author}{Steve \surnamestart Kremer\surnameend} \&
  \bibinfo{author}{Ludovic \surnamestart Robin\surnameend}
  (\bibinfo{year}{2017}): \emph{\bibinfo{title}{Formal Verification of
  Protocols Based on Short Authenticated Strings}}.
\newblock In: {\sl \bibinfo{booktitle}{30th {IEEE} Computer Security
  Foundations Symposium, {CSF} 2017, Santa Barbara, CA, USA, August 21-25,
  2017}}, pp. \bibinfo{pages}{130--143}, \doi{10.1109/CSF.2017.26}.

\bibitemdeclare{incollection}{DBLP:books/el/RV01/DershowitzP01}
\bibitem{DBLP:books/el/RV01/DershowitzP01}
\bibinfo{author}{Nachum \surnamestart Dershowitz\surnameend} \&
  \bibinfo{author}{David~A. \surnamestart Plaisted\surnameend}
  (\bibinfo{year}{2001}): \emph{\bibinfo{title}{Rewriting}}.
\newblock In: {\sl \bibinfo{booktitle}{Handbook of Automated Reasoning}},
  \bibinfo{publisher}{Elsevier and {MIT} Press}, pp. \bibinfo{pages}{535--610},
  \doi{10.1016/B978-044450813-3/50011-4}.

\bibitemdeclare{article}{FiazzaPV15}
\bibitem{FiazzaPV15}
\bibinfo{author}{Maria{-}Camilla \surnamestart Fiazza\surnameend},
  \bibinfo{author}{Michele \surnamestart Peroli\surnameend} \&
  \bibinfo{author}{Luca \surnamestart Vigan{\`{o}}\surnameend}
  (\bibinfo{year}{2015}): \emph{\bibinfo{title}{Defending Vulnerable Security
  Protocols by Means of Attack Interference in Non-Collaborative Scenarios}}.
\newblock {\sl \bibinfo{journal}{Front. {ICT}}} \bibinfo{volume}{2015},
  \doi{10.3389/fict.2015.00011}.

\bibitemdeclare{inproceedings}{Lowe99}
\bibitem{Lowe99}
\bibinfo{author}{Mei~Lin \surnamestart Hui\surnameend} \&
  \bibinfo{author}{Gavin \surnamestart Lowe\surnameend} (\bibinfo{year}{1999}):
  \emph{\bibinfo{title}{Safe Simplifying Transformations for Security
  Protocols}}.
\newblock In: {\sl \bibinfo{booktitle}{Proceedings of the 12th {IEEE} Computer
  Security Foundations Workshop, {CSFW} 1999, Mordano, Italy, June 28-30,
  1999}}, \bibinfo{publisher}{{IEEE} Computer Society}, pp.
  \bibinfo{pages}{32--43}, \doi{10.1109/CSFW.1999.779760}.

\bibitemdeclare{inproceedings}{ASIACCS}
\bibitem{ASIACCS}
\bibinfo{author}{Zhiwei \surnamestart Li\surnameend} \&
  \bibinfo{author}{Weichao \surnamestart Wang\surnameend}
  (\bibinfo{year}{2012}): \emph{\bibinfo{title}{Towards the attacker's view of
  protocol narrations (or, how to compile security protocols)}}.
\newblock In \bibinfo{editor}{Heung~Youl \surnamestart Youm\surnameend} \&
  \bibinfo{editor}{Yoojae \surnamestart Won\surnameend}, editors: {\sl
  \bibinfo{booktitle}{7th {ACM} Symposium on Information, Computer and
  Communications Security, {ASIACCS} '12, Seoul, Korea, May 2-4, 2012}},
  \bibinfo{publisher}{{ACM}}, pp. \bibinfo{pages}{44--45},
  \doi{10.1145/2414456.2414481}.

\bibitemdeclare{article}{lowe-nspk}
\bibitem{lowe-nspk}
\bibinfo{author}{Gavin \surnamestart Lowe\surnameend} (\bibinfo{year}{1996}):
  \emph{\bibinfo{title}{Breaking and Fixing the Needham-Schroeder Public-Key
  Protocol Using {FDR}}}.
\newblock {\sl \bibinfo{journal}{Software - Concepts and Tools}}
  \bibinfo{volume}{17}(\bibinfo{number}{3}), pp. \bibinfo{pages}{93--102},
  \doi{10.1007/3-540-61042-1_43}.

\bibitemdeclare{article}{NS78}
\bibitem{NS78}
\bibinfo{author}{Roger~M. \surnamestart Needham\surnameend} \&
  \bibinfo{author}{Michael~D. \surnamestart Schroeder\surnameend}
  (\bibinfo{year}{1978}): \emph{\bibinfo{title}{Using Encryption for
  Authentication in Large Networks of Computers}}.
\newblock {\sl \bibinfo{journal}{Commun. {ACM}}}
  \bibinfo{volume}{21}(\bibinfo{number}{12}), pp. \bibinfo{pages}{993--999},
  \doi{10.1145/359340.359342}.

\bibitemdeclare{phdthesis}{these-robin}
\bibitem{these-robin}
\bibinfo{author}{Ludovic \surnamestart Robin\surnameend}
  (\bibinfo{year}{2018}): \emph{\bibinfo{title}{V{\'{e}}rification formelle de
  protocoles bas{\'{e}}s sur de courtes chaines authentifi{\'{e}}es. (Formal
  verification of protocols based on short authenticated strings)}}.
\newblock Ph.D. thesis, \bibinfo{school}{University of Lorraine, Nancy,
  France}.
\newblock \urlprefix\url{https://tel.archives-ouvertes.fr/tel-01767989}.

\bibitemdeclare{inproceedings}{machiavelli}
\bibitem{machiavelli}
\bibinfo{author}{Paul~F. \surnamestart Syverson\surnameend} \&
  \bibinfo{author}{Catherine~A. \surnamestart Meadows\surnameend}
  (\bibinfo{year}{2000}): \emph{\bibinfo{title}{Dolev-Yao is no better than
  Machiavelli}}.
\newblock In: {\sl \bibinfo{booktitle}{Proceedings of the first {\em Workshop
  on Issues in the Theory of Security (WITS'00)}}}, \doi{10.21236/ADA464936}.

\bibitemdeclare{inproceedings}{strand}
\bibitem{strand}
\bibinfo{author}{F.~Javier \surnamestart Thayer\surnameend},
  \bibinfo{author}{Jonathan~C. \surnamestart Herzog\surnameend} \&
  \bibinfo{author}{Joshua~D. \surnamestart Guttman\surnameend}
  (\bibinfo{year}{1998}): \emph{\bibinfo{title}{Strand Spaces: Why is a
  Security Protocol Correct?}}
\newblock In: {\sl \bibinfo{booktitle}{Security and Privacy - 1998 {IEEE}
  Symposium on Security and Privacy, Oakland, CA, USA, May 3-6, 1998,
  Proceedings}}, \bibinfo{publisher}{{IEEE} Computer Society}, pp.
  \bibinfo{pages}{160--171}, \doi{10.1109/SECPRI.1998.674832}.

\end{thebibliography}

\appendix

\section{Relation with the notion of static equivalence}

The notions of equivalence \textit{wrt} the refinement relation and
static equivalence are strongly related. The different setting is
justified by the different handling of nonces: in~\cite{AbadiC06}
contexts can contain any constant, so the secret constants in a trace
have to be protected using \(\pi\)-calculus'
\(\nu\)
operator, while we disallow non-public constants in contexts, which
means that no constants can be used but the ones in the deduction
system or those published (explicitly or implicitly) in the sequence
of messages. We prove in Theo.~\ref{theo:decidability} that the two
notions are identical modulo the generation of new constants with the
deduction system \call{N} of Ex.~\ref{ex:nonces}.

The \textit{AttackDetectability} problem defined in this paper is new,
but it is strongly related to the static equivalence problem. In order to
show this relation, let us introduce \emph{frames}, which are strands
with a hidden set of constants.

\begin{definition}{(Frames,~\cite{AbadiC06})}
  A \emph{frame} is a couple \((\tilde n,s)\) where \(\tilde n\) is a set of
  constants and \(s\) is a positive trace, and is usually denoted \(\nu
  \tilde n.s\).
\end{definition}

Technically the definition in~\cite{AbadiC06} replaces \emph{positive
  trace} \(s\)
by a substitution of domain \(x_1,\ldots,x_n\)
that we have noted \(\sigma_s\).
The application of a context \(C\)
on a frame \(\nu\tilde n.s\)
is equal to the application of \(C\)
on \(s\).
We are now ready to define the static equivalence problem for a
deduction system \(\mathcal{D}\).

\begin{decisionproblem}{\ensuremath{\mathop{\it StaticEquivalence}_{\mathcal{D}}(\varphi_1,\varphi_2)}}
  \begin{description}
  \item[Input:] Two frames \(\varphi_i = \nu \tilde n_i.s_i\) for \(i=1,2\)
  \item[Output:] \textsc{yes} if the frames have an equal length and
    for all pair \(C_1,C_2\) of public contexts and all function
    \(\theta: \Var{C_1}\cup\Var{C_2}\to \set{x_1,\ldots,x_n}\cup
    (\Constants \setminus (\tilde n_1 \cup\tilde n_2))\) we have \(
    C_1\theta \sigma_{s_1}  = C_2\theta \sigma_{s_1} \) if, and only if, \(
    C_1\theta \sigma_{s_2} =_{\mathcal{E}} C_2\theta \sigma_{s_2} \).
  \end{description}  
\end{decisionproblem}

Attack detectability is related to static equivalence with the following theorem:

\begin{trivlist}
\item \textbf{Theorem~\ref{theo:decidability}.}~\sl
  Let \(\mathcal{D}\)
  be a deduction system, and \(\mathcal{N}\)
  be the deduction system of Ex.~\ref{ex:nonces}.
  Then \ensuremath{\mathop{\it
      AttackDetectability}_{\mathcal{D}\cup\mathcal{N}}} on strands
  that do not contain symbols of \(\Constants_{\mathcal{N}}\) is
  polynomial-time reducible to \ensuremath{\mathop{\it
      StaticEquivalence}_{\mathcal{D}}}.
\end{trivlist}

\begin{proof}
   Given a trace \(s\)
  we let \(\varphi_s=\nu \Const{s}.[s_1,\ldots,s_n]\).
  This construction is clearly polynomial time. Let \(t_1,t_2\)
  be the two traces in the  presentation of the attack \(\mathcal{A}\)
  on \(I_P\) to \(M\).

  First, if \(t_1\) and \(t_2\) are of different length or have different label sequence, one can
  respond to the \ensuremath{\mathop{\it AttackDetectability}} in
  polynomial time. So let us assume the two strands have the same
  length and the same label sequence. Also, we assume that \(t_1,t_2\)
  do not contain the symbols of the \(\mathcal{N}\) deduction system.

  Let us prove that \(t_1\)
  and \(t_2\)
  are discernable wrt the deduction system
  \(\mathcal{D}\cup\mathcal{N}\)
  if, and only if, the frames \(\varphi_{s_1},\varphi_{s_2}\)
  are not statically equivalent wrt the deduction system
  \(\mathcal{D}\).

  First let us assume that the attack presentation \((t_1,t_2)\)
  is detectable, and wlog assume that \(t_1\) does not refine \(t_1\)
  for the deduction system \(\mathcal{D}\cup\mathcal{N}\).
  Thus there exists two \(\mathcal{D}\cup\mathcal{N}\)-contexts
  such that
  \(C_1 t_1 =_{\mathcal{E}} C_2 t_1\) but \(C_1 t_2
  \not=_{\mathcal{E}} C_2 t_2\). Since we assume that constants
  occurring in \(s_1,s_2\) are away from \(\Constants_{\mathcal{N}}\),
  we construct \(C_1',C_2'\) and \(\theta'\) as follows. For each
  constant \(c\in\Constants_{\mathcal{N}}\) occurring in \(C_1\) or \(C_2\):
  \begin{itemize}
  \item replace in the contexts \(c\) with a new variable \(x_c\);
  \item define \(\theta\) as follows:
    \[
      \theta(x)=\left\lbrace
        \begin{array}[c]{ll}
          x_i & \text{ if } x \in\Var{C_1}\cup\Var{C_2}\\
          c & \text{ if } x =x_c\\
        \end{array}
        \right.
    \]
  \end{itemize}
  By construction \(C_1',C_2'\) are \(\mathcal{D}\)-public
  contexts and \(\theta\) maps each variable of these contexts to
  either a variable or to a constant away from \(s_1,s_2\).
  Thus \(C_1',C_2'\) and \(\theta'\) are witnesses that the frames
  \(\varphi_s\) and \(\varphi_{s'}\) are not \(\mathcal{D}\)-statically equivalent.

  Conversely assume that the two frames
  \(\varphi_{s_1},\varphi_{s_2}\) are not statically equivalent. Then
  there exist \(\mathcal{D}\) contexts \(C_1,C_2\) and \(\theta:\Var{C_1}\cup\Var{C_2}\to
  \Constants \setminus (\Const{s_1}\cup\Const{s_2})\) such that wlog
  \(C_1\theta \varphi_{s_1}=_{\mathcal{E}} C_2\theta \varphi_{s_1}\).
  Replacing each free constant \(c\) by a constant in
  \(\Constants_{\mathcal{N}}\) yields appropriate contexts for the
  \(\mathcal{D}\cup \mathcal{N}\) attack detectability.
\end{proof}

\end{document}